\documentclass[a4paper,UKenglish]{llncs}
\usepackage{microtype}
\usepackage{amsmath}
\usepackage{amsfonts}
\usepackage{amssymb}
\usepackage{color}
\usepackage[margin=1.8in]{geometry}
\usepackage{mathrsfs}
\usepackage{tikz}
\usepackage{multirow} 
\usepackage{slashbox}
\usetikzlibrary{arrows}

\newtheorem{define}{Definition}

\newcommand\motnouv[1]{\emph{#1}}

\newcommand\N{\mathbb{N}}
\newcommand\Z{\mathbb{Z}}

\newcommand\Part{\mathscr{P}}

\def\E{\mathop{\mathbb{E}}\nolimits}
\def\ETf{\mathop{\mathbb{E}(T_f)}\nolimits}

\newcommand\I{{I}}

\begin{document}

\mainmatter

\title{Global Versus Local Computations:
Fast Computing with Identifiers}

\author{Mika\"el RABIE}


\institute{LIP, ENS de Lyon\\  69007 Lyon, France\\
  \texttt{mikael.rabie@ens-lyon.org}   }



\maketitle


\begin{abstract}
This paper studies what can be computed by using probabilistic
local interactions with agents with a very restricted power in polylogarithmic
parallel time.

It is known that if agents are only finite state (corresponding to the Population
Protocol model by Angluin {\em et al.}), then only semilinear
predicates over the global input can be computed. In fact,  if 
the population starts with a unique leader, these predicates
can even be computed in a polylogarithmic parallel time.

If identifiers are added (corresponding to the Community Protocol model by
Guerraoui and Ruppert), then more global predicates over the input
multiset can be computed.
Local predicates over the input sorted according to the identifiers
can also be computed, as long as the identifiers are ordered.
The time of some of those predicates might require exponential parallel time.

In this paper, we consider what can be computed with Community Protocol
in a polylogarithmic number of parallel interactions. We introduce the  class
 $CPPL$ corresponding to protocols that use $O(n\log^kn)$, for some $k$,
expected interactions to compute their predicates, 
or equivalently a polylogarithmic number of parallel expected interactions.

We provide some computable protocols, some boundaries of the class,
using the fact that the population can compute its size.
We also prove two impossibility results providing some arguments showing that local computations are no
longer easy: the population does not have the time to compare a linear number
of consecutive identifiers. The {\em Linearly Local} languages,
such that the rational language $(ab)^*$,
are not computable.




\end{abstract}

\section{Introduction}


Population Protocols, introduced by Angluin \emph{et al.}
in 2004 \cite{AspnesADFP2004}, corresponds to a model of finite states devices
with a very restricted memory using pairwise interactions to communicate and
compute a global result.
Predicates computable by  population protocols
have been characterized as being precisely the semi-linear predicates;
i.e. those equivalent to be definable in
first-order Presburger arithmetic \cite{angluin2007cpp,AspnesADFP2004}.  Semi-linearity
was shown to be sufficient, and necessary. Those predicates use the global multiset of the input.

Later works on population protocols have concentrated on characterizing
what predicates on the input configurations can be stably computed in
different variants of the models and under various assumptions.
Variants of the original model considered so far include restriction
to one-way communications~\cite{angluin2007cpp}, restriction to
particular interaction graphs~\cite{AngluinACFJP2005}.  Various kinds of fault
tolerance have been studied for population
protocols~\cite{Delporte-GalletFGR06}, including the search for
self-stabilizing solutions~\cite{AngluinAFJ2005}.  
Some works also include the 
Probabilistic Population Protocol model that makes a random scheduling
assumption for interactions~\cite{DBLP:journals/dc/AngluinAE08a,DBLP:conf/wdag/DotyS15}.

Some works extend this model. The edges of the 
interaction graph may have  states that belong to a constant-size 
set. This model called the \emph{mediated population protocol} is 
presented in \cite{MichailCS11}. 
The addition on Non-Determinism has been studied in \cite{DBLP:conf/opodis/BeauquierBRR12}.
The research of Self-Stabilization (over some fairness assumption)
has been explored in \cite{AngluinAFJ2005,DBLP:conf/sss/BeauquierBK09}.
An extension with sensors offering a {\em cover-time} notion was also
studied in~\cite{DBLP:conf/sss/BeauquierBBD11}.
A recent study in \cite{DBLP:conf/mfcs/MertziosNRS16} also focused on finding the median agent in an extension 
of the model called Arithmetic Population Protocols.

More generally, the population protocol model shares many features with other models 
already considered in the literature. In particular, models of 
pairwise interactions have been used to study the propagation of 
diseases~\cite{Heth00}, or rumors \cite{dk65}. In chemistry, the 
chemical master equation has been justified using (stochastic) 
pairwise interactions between the finitely many molecules
\cite{gillespie1992rdc,Murray-VolI}. 
The variations over the LOCAL model \cite{DBLP:conf/spaa/FraigniaudKL07}
 can be seen as a restriction over
the interactions (using a graph) but with a set of possible improvements in
agents' capacities.

Agents have been endowed with even stronger tools in different models.
The \emph{passively mobile
 protocols} introduced by Chatzigiannakis \emph{et
  al.}  \cite{PMC} constitutes a
generalization of the population protocol model where finite state
agents are replaced by agents that correspond to arbitrary Turing
machines with $O(S (n))$ space per-agent, where $n$ is the number
of agents. As agents remain initially anonymous, only functions over
the global input can be computed.

The \emph{community protocols} introduced by Guerraoui and
Ruppert~\cite{guerraoui2009names} are closer to the original
population protocol model, assuming \textit{a priori} agents with
individual very restricted computational capabilities.  In this model,
agents are no longer anonymous: each agent has a unique
identifier and can only remember $O(1)$ other agent
identifiers.  Guerraoui
and Ruppert~\cite{guerraoui2009names} using results about the
so-called storage modification machines~\cite{schonhage1980storage},   proved
that such protocols 
simulate  Turing machines:  Predicates computed by this
model with $n$ agents are precisely the predicates in $NSPACE(n \log n)$. 
The sorted input symbols according to the identifiers can be analysed locally
by the protocols to compute the right output.
In~\cite{Rabie15},
the possibility that identifiers are no longer unique is explored through the
{\em homonym population protocols} model.

\subsection*{Motivation}

Angluin {\em et al.}, in \cite{DBLP:journals/dc/AngluinAE08a},
  prove that any computable predicate by a Population Protocol can be computed
in $O(n\log^5n)$ expected interactions, as long as there is a unique leader at
the beginning. This article includes some arguments leading
to the idea that there might exist protocols computing a leader election
in $O(n\log n)$ expected interactions. Doty {\em et al.} proved in \cite{DBLP:conf/wdag/DotyS15}
that there cannot be a protocol computing a leader so fast. They proved that a protocol needs $\Omega(n^2)$
expected interactions to get to a configuration with a single leader, if every agent is a potential
candidate at the beginning.

The exact characterization of what can be computed by populations
having unique leaders  gave the motivation to look to what can be computed
in $O(n\log^kn)$ expected interactions (for any $k>0$), with the Community
Protocols model \cite{guerraoui2009names}.
We consider, as in \cite{DBLP:journals/dc/AngluinAE08a}, that each pair of agents (or identifiers) have the
same probability to be chosen at each step of a computation. In \cite{DBLP:journals/dc/AngluinAE08a},
it is considered that dividing the number of expected interaction by $n$ provides
the expected number of parallel interactions.

Community protocols can be seen as interactions controlled by devices
in a social group. For example, identifiers can correspond to phone numbers,
and the devices can be applications on smartphones.
In this vision, it seems intuitive to consider that a group of individuals do not want
to stay too long together to compute some global information.
Sorting a group of people depending on phone numbers
to look for patterns does not seem intuitive, and hence useful.\\




This paper introduces the class \motnouv{CPPL}, corresponding
to what can be computed with Community Protocols in a
polylogarithmic number of expected parallel interactions
(which corresponds to a number of expected interactions 
bounded above by $n\log^k n$ for some $k$), or a polylogarithmic
number of epidemics or broadcasts. We introduce some protocols, proving that
the size of the population (or some subset) can be computed
in some sense to be explained.

We then show the weakness of this model based on the fact that
local computation cannot be performed over the whole input.
More precisely, we prove that only a polylogarithmic number of agents can
find the next or previous identifier to their own. 
We also introduce the class of
linearly local languages, containing the rational language $(ab)^*$, 
and prove that none of its elements cannot be computed.

We finish with some comparisons with other computational classes.
We introduce a class of Turing Machine trying to match the expressive power of $CPPL$.
Those machines use a polylogarithmic space of computation,
and is able to use the tools we found. This machine has access to global
informations of the input, but can focus locally only on a polylogarithmic number
of regions of the input.

The paper is organized as follows:
Section \ref{sec:1} provides the Community Protocol model introduced in \cite{guerraoui2009names}
and includes some examples.
Section \ref{sec:2} provides some elements and results about fast computing
with Population Protocols from \cite{DBLP:journals/dc/AngluinAE08a}.
Section \ref{sec:3} explains a way to keep the fairness of our protocols and describes
a way to compute the size of the population.
Section \ref{sec:4} introduces the notion of Linearly Local languages
and proves that these languages are not in $CPPL$.
Section \ref{sec:5} provides some complexity bounds on the class
$CPPL$.

\section{Model}\label{sec:1}

We present now the model introduced by
Guerraoui and Ruppert in \cite{guerraoui2009names}: Agents have
unique identifiers, and can store a fixed number of them. Agents
can compare 2 identifiers.
We consider that, unlike in \cite{Rabie15},
agents cannot know when two identifiers are consecutive.

This model has been proved in \cite{guerraoui2009names} to correspond to $NSPACE(n\log n)$,
even when we add a fixed number of byzantine agents.
We will not consider byzantine agents in this paper.


\begin{define} A \motnouv{Community Protocol} is given by seven elements $(U,B,d,\Sigma,\iota,\omega,
\delta)$ where:
\begin{itemize}
\item $U$ is the infinite ordered set of identifiers.
\item $B$ is a finite set of basic states.
\item $d \in\N$  is the number of identifiers that can be
  remembered by an agent.
\item $\Sigma$ is the finite set of entry symbols.
\item $\iota$ is an input function $\Sigma\rightarrow B$.
\item $\omega$ is an output function $B\rightarrow \{True,False\}$.
\item $\delta$ is a transition function $Q^2\rightarrow Q^2$, with
  $Q=B\times U \times (U  \cup \{\_\})^d$.
\end{itemize}
The set $Q=B\times U \times (U  \cup \{\_\})^d$
of possible states each agents can have is such
that each agent carries three elements: its identifier, its state, and $d$ slots
for identifiers.

The transition function $\delta$ has two restrictions: Agents cannot
store identifiers that they never heard about, and the transitions must
only depend on relative position of the identifiers in the slots and on
the state in $B$. More formally, we have:
\begin{enumerate}
\item if $\delta(q_1,q_2)=(q'_1,q'_2)$, and $id$ appears in
$q'_1$ or $q'_2$ then $id$ must appear in $q_1$ or in $q_2$.
\item whenever $\delta(q_1,q_2)=(q'_1,q'_2)$, 
let  $u_1< u_2< \dots<u_k$ be the distinct  identifiers that appear in any of
the four states $q_1,q_2,q'_1,q'_2$. Let $v_1< v_2< \dots< v_k$
be distinct identifiers. If $\rho(q)$ is the state obtained from $q$ by replacing all
occurrences of each  identifier $u_i$ by $v_i$, then we require that
$\delta(\rho(q_1),\rho(q_2))=(\rho(q'_1),\rho(q'_2))$. 
\end{enumerate}
\end{define}
We also add the fact that $\delta$ cannot change the identifier of an agent.

As a convention,
we will often call an agent of initial identifier $id\in U$ the agent $id$.
We will sometimes write $Id_k$ for the $k$th identifier present in the population.
An agent with identifier $id$, in state $q$ and with a list of $d$ identifiers
$L=id_1$, \ldots, $id_d$ will be written in what follows 
 $q_{id,id_1,\ldots,id_d}$.

\begin{example}[Leader Election]\label{lecp}
It is possible to compute a Leader Election (a protocol where all agents start in state $L$
from which we want to reach a configuration with a single $L$: the leader), where the leader will be the agent
with the smallest identifier, with $O(n\log n)$ expected interactions. As a reminder,
without identifiers, a protocol needs $\Omega(n^2)$ expected interactions to elect a leader \cite{DBLP:conf/wdag/DotyS15}.

Agents will store the identifier of their leader.
Here is the protocol, using above notations for rules:
\begin{itemize}
\item $B=\{L,N\}$.
\item $d=1$.
\item $\Sigma=L$, $\iota(L)=L$ and $\omega(L)=\omega(N)=True$.
\item $\delta$ is such that the non-trivial rules (i.e. where at least one state changes) are:
\end{itemize}
\begin{center}
\begin{tabular}{ r @{\hspace{0,2cm}} l @{$\rightarrow$} r @{\hspace{0,2cm}}  l l}
$L_{id_a,\_}$ & $L_{id_b,\_}$ & $L_{id_a,\_}$ & $N_{id_b,id_a}$ & with $id_a<id_b$\\
$L_{id_a,\_}$ & $N_{id_b,id_c}$ & $L_{id_a,\_}$ & $N_{id_b,id_a}$ & with $id_a<id_c$\\
$L_{id_a,\_}$ & $N_{id_b,id_c}$ & $N_{id_a,id_c}$ & $N_{id_b,id_c}$ & with $id_c<id_a$\\
$N_{id_a,id_b}$ & $N_{id_c,id_d}$ & $N_{id_a,id_b}$ & $N_{id_c,id_b}$ & with $id_b<id_d$\\
\end{tabular}
\end{center}

To determine the speed of this protocol, it suffices to realize that the final leader actually
does an epidemic to spread its identifier (epidemic is defined in Definition \ref{epidemicDef}).
An epidemic takes $O(n\log n)$ expected
interactions. Thus, the leader election can be performed in $O(n\log n)$ expected
interactions.
The notions of time and computation are defined in what follows.
\end{example}

\begin{remark}
To ensure that at some point, a single leader remains in the population, Gerraoui
{\em et al.} uses the notion of Fairness introduced in the Population Protocols model \cite{AspnesADFP2004}.
As we work here with probabilistic interaction (each pair of agents has the same
probability to interact), the fairness notion will not be needed.
\end{remark}

\begin{define}
An \motnouv{Input} is a subset of $U\times\Sigma$ such that any element of $U$
(the elements of $U$ being called \motnouv{Identifiers})  can
appear at most once. Inputs will often be seen as words of $\Sigma^*$,
as it is possible to sort the input elements according to the identifiers
(recall that we consider that $U$ is ordered).
An input $u=u_1\ldots u_n$ is such that the agent with the smallest identifier
has input $u_1$, the second has input $u_2$\ldots

The \motnouv{Initial State} of an agent assigned with the identifier
$id$ and the input $s$ is $(\iota(s),id,\_^d)$, where $\_^d$ states for $d$
repetitions of the empty slot $\_$.

A \motnouv{Configuration} is a subset of $Q$ where two elements cannot have the same first
identifier (i.e. two agents must have two distinct identifiers). 

A \motnouv{Step} is the transition between two configurations
$C\rightarrow C'$, where only two agents' states may change: we
apply to the two agents $a_1$ and $a_2$ the rule
corresponding to their respective state $q_1$ and $q_2$, i.e. if
$\delta(q_1,q_2)=(q'_1,q'_2)$ (also written by rule  $q_1$ $
q_2\rightarrow q'_1$ $q'_2$), then in $C'$ the respective states of $a_1$
and $a_2$ are $q'_1$ and $q'_2$. All other agents have the same state
in $C$ and $C'$.

A configuration has an \motnouv{Output} $y\in \{True,False\}$
if for each state $b\in B$ present in the population, $\omega(b)=y$.
A configuration $C$ is said \motnouv{Output Stable} if it has an output $y$ and
 if, for any $C'$ reachable from $C$, $C'$ has also the output $y$.

An input $w\in\Sigma^*$ has an \motnouv{Output} $y\in Y$ if from any
reachable configuration from the initial configuration, 
we can reach an output stable configuration of output $y$.
It means that from the input, the protocol will reach with probability 1
an output stable configuration, and there is a single output $y$ reachable.
The input is \motnouv{Accepted} if and only if it has output $True$.

A protocol  \motnouv{Computes} a set $L$ if,
  for any input word $w\in\Sigma^*$, the protocol provides an output, and
the protocol accepts $w$ if and only if $w\in L$.
We then say that $L$ is \motnouv{Computable}.
We will sometimes say that the protocol is \motnouv{Las Vegas}, as it will always succeed to
provide an output with probability 1.

A language is \motnouv{Computed in $f(n)$ Expected Interactions} if, for any input $w$,
the expected number of interactions to reach an output stable configuration
is bounded above by $f(|w|)$.
\end{define}

%
%
%
%

The Community Protocols model has been fully characterized:

\begin{theorem}[\cite{guerraoui2009names}]\label{cpmain}
The decisions problems computable by community protocols correspond
exactly to the class $NSPACE(n\log n)$.

The set of languages computable  by community protocols is
$NSPACE(n\log n)$.
\end{theorem}


Let us introduce now the class we will work with in this paper:

\begin{define}
We define the class $\motnouv{CPPL}$ as 
the sets of languages that can be
recognized by a Community Protocol with $O(n\log^kn)$ expected interactions for some $k\in\N$,
where each pair of agents has the same probability to interact at each moment.

We say that a function $f$ is \motnouv{n-polylog} if there exists some $k$ such that we have
$f(n)\le n\log^kn$.
\end{define}

\section{Fast Computing Known Results}\label{sec:2}

We introduce here some of the elements and results in \cite{DBLP:journals/dc/AngluinAE08a} by Angluin {\em et al.}.
These elements are on the Population Protocols model. It corresponds to
the case where agents do not have identifiers.

The results are based under the assumption that the population starts with a unique leader.
With community protocols, this assumption will no longer be used,
we will always consider the leader to be the agent with the smallest identifier (see Example \ref{lecp}).

We introduce the main result and some tools from \cite{DBLP:journals/dc/AngluinAE08a} that will
be used in this paper. We first introduce the notion of epidemics, which will
be our main tool to perform computations. We will quickly talk about
the Phase Clock Protocol that permits to be sure with high probability that
an epidemic had the time to happen. We finish with a complexity result.

\subsection{Epidemics}

The epidemic is the most important probabilistic protocol. Its purpose
is to spread or gather information. It will permit for example to get an identifier,
to check the state of an agent of a given identifier, to check if there exists
some agent in a given state\ldots

The important element with this tool is that an epidemic takes $O(n\log n)$
expected interactions. Intuitively, in parallel, at each step, the number
of agents aware of the epidemic doubles, using $O(\log n)$ parallel steps
to spread.

\begin{define}[\cite{DBLP:journals/dc/AngluinAE08a}]\label{epidemicDef}
An \motnouv{Epidemic Protocol} is a protocol who spreads some information through an epidemic.
The purpose is, for a leader, to \motnouv{Infect} each agent.
More formally, if 0 represents the not infected state and 1
the infected one, there is just a non trivial rule:
\begin{center}
\begin{tabular}{ r @{\hspace{0,2cm}} l @{$\rightarrow$} r @{\hspace{0,2cm}}  l}
$1$ & $0$ & $1$ & $1$\\
\end{tabular}
\end{center}

Most of the time, it will be a leader who will start a spreading of some information.
The computation will start in the configuration $10^{n-1}$ (one agent in state $1$, the others
being in state $0$), where
1 represents the leader.
\end{define}

\begin{proposition}[\cite{DBLP:journals/dc/AngluinAE08a}]\label{epidemy}
Let $T$ be the expected number of interactions before an epidemic protocol
starting with a single infected agent infects all the other ones. For any fixed
$c>0$, there exist positive constants $c_1$ and $c_2$
such that, for sufficiently large $n$, with probability at most $1-n^{-c}$:
$$c_1n\log n\le T\le c_2n\log n$$
\end{proposition}

From this theorem, we know that
any epidemic protocol will take $\Theta(n\log n)$ expected interactions.
If we are (almost) sure that more than
$c_2n\log n$ interactions occurred, we will be (almost) sure that an epidemic has finished.

To be almost sure that at least $c_2n\log n$ interactions have occurred,
\cite{DBLP:journals/dc/AngluinAE08a} introduced the Phase Clock Protocol.
 The leader runs a clock
between $0$ and $m$ for some $m>0$.
Each agent tries to store the current time, 
 following some updating rules.
Each time the clock loops (i.e. the leader reaches $m$ and resets the clock), 
the population is almost sure
that at least $c_2n\log n$ interactions have occurred.

\begin{proposition}[\cite{DBLP:journals/dc/AngluinAE08a}]
For any fixed $c,d_1>0$, there exist two constants $m$ and $d_2$ such that,
for all sufficiently large $n$, with probability at least $1-n^{-c}$ the
phase clock protocol with parameter $m$,
completes $n^c$ rounds,
 where the minimum number of
interactions in any of the $n^c$ rounds is at least $d_1n\log n$ and the
maximum is at most~$d_2n\log n$.
\end{proposition}

This result permits to be sure, with high probability, that for $n^c$ rounds,
in each round, an epidemic had the time to occur.

\subsection{Presburger's Arithmetic}\label{subsectfcpp}

The main result from \cite{DBLP:journals/dc/AngluinAE08a} is that, if the population starts
with a unique leader, any computable predicate by a population protocol
can be computed with $O(n\log^5n)$ expected interactions.

\begin{theorem}[\cite{DBLP:journals/dc/AngluinAE08a}]\label{thfc}
For any predicate $P$ definable in Presburger's Arithmetic, and for any $c>0$, there is a probabilistic
population protocol with a leader to compute $P$ without error that converges
in $O(n\log^5n)$ interactions with probability at least $1-n^{-c}$.
\end{theorem}

As a reminder, those predicates  correspond
to boolean combinations of:
\begin{itemize}
\item \motnouv{Threshold Predicate}: $\sum a_ix_i\ge b$, with $a_1,\ldots,a_n,b\in\Z^{n+1}$.
\item \motnouv{Modulo Predicate}: $\sum a_ix_i\equiv b[c]$, with $a_1,\ldots,a_n,b,c\in\Z^{n+2}$.
\end{itemize}
where $x_i$ corresponds to the number of agents with input $i\in\Sigma$.
This also corresponds to semilinear sets.

\begin{corollary}
Any predicate definable in Presburger's Arithmetic is in $CPPL$.
\end{corollary}
\begin{proof}We use the two following facts:
\begin{itemize}
\item The Leader Election can be performed in $O(n\log n)$ (see Example \ref{lecp}).
\item Any predicate definable in Presburger's Arithmetic can be computed in $O(n\log^5 n)$ expected
interactions (see Theorem \ref{thfc}), as long as there is a single leader.
\end{itemize}
Each agent stores the smallest identifier it has heard about in its Leader slot.
It links its internal clock to the leader: if it meets an agent storing a smallest identifier, it acts as if its own clock
was at 0, and performs the interaction with the other agent accordingly.
Hence, each agent will act as in the protocols of
\cite{DBLP:journals/dc/AngluinAE08a} as soon as it hears about the right leader's identifier
(in \cite{DBLP:journals/dc/AngluinAE08a}, agents start their role in the computation as soon as they get
instruction from the leader, or from someone who transmits leader's instruction through an epidemic).
\end{proof}

\section{Some Computable Protocols}\label{sec:3}

We are now able to introduce some probabilistic protocols, including a
complex one that encodes the size of the population.
Let first introduce the following notion:

\begin{define}\label{deffindnext}
We will often talk about \motnouv{Next} and \motnouv{Previous}.
Those are two functions $U\to U$ that provides, to a given identifier, the next one/previous
one present in the population.
More formally:
\begin{itemize}
\item $Next({id_a})=\min\{{id_b}:{id_b}>{id}_a\}$.
\item $Previous({id_a})=\max\{{id_b}:{id_b}<{id}_a\}$.
\end{itemize}
By convention, $Next$ of the highest identifier is the smallest, and $Previous$
of the smallest identifier is the highest one. Thus, these two functions are bijective.

Sometimes, $Next$ and $Previous$ will be slots in protocols,
with the purpose to find the right identifier corresponding to the function.
"Finding its $Next$" means that the agent needs to put the right
identifier in its slot $Next$.
\end{define}

\subsection{From Monte Carlo to Las Vegas Protocols}

We considered in the previous section Monte Carlo protocols (i.e. protocols having eventually a probability of failure).
 We accept that the protocols might have some probability
of failure, as long as we can minimize it as much as needed (we use the same bound of
$1-n^{-c}$ as in \cite{DBLP:journals/dc/AngluinAE08a}). Those protocols alone do not compulsory compute any set.

We provide in this paper Monte Carlo descriptions of the protocols.
We consider that the protocols also run in parallel a Las Vegas protocol
 providing the right output with
probability 1 (the corresponding Las Vegas protocols exist, as a consequence of Theorem~\ref{cpmain}).
 The protocol detects, as in \cite{DBLP:journals/dc/AngluinAE08a}, when
the Las Vegas protocol should have finished to find the output. At this point, each agent switches its
output from the Monte Carlo protocol's to the Las Vegas protocol's. With probability at least $1-n^{-c}$,
this will not change the output.

Here is a small result to justify that we can transform our protocols presented in this paper
in Las Vegas ones by multiplying the expected number of interactions by $n^3$:

\begin{proposition}\label{MonteCarlo}
Let be a population where all agents has found their $Next$ (see definition~\ref{deffindnext}).
There exists a protocol that simulates an epidemic spread from an agent
taking $O(n^3)$ expected interactions, with a success of probability 1. In the new protocol, the agent meets
all the other ones in the population.

\end{proposition}

\begin{proof}
We suppose that all agents have already found their $Next$, and we suppose all agents know the
leader's identifier, being the smallest identifier in the population. The agent needs to meet the leader,
then remembers the $Next$ of the leader, meets it, remembers its $Next$\ldots~until it finds the agent
having as $Next$ the leader's identifier. At this point, the agent has met all agents in the population.

Each step takes $\frac{n(n-1)}2$ expected interactions, and we have $n$ steps. Hence,
this protocol takes $O(n^3)$ expected interactions to derandomize the epidemic from the initial agent.
\end{proof}

Finding each $Next$ needs at most $O(n^2\log n)$ expected interactions (which corresponds to the
number of interactions expected before every possible interaction has occurred at least once).
Detecting when an agent found a new $Next$ is easy: the corresponding agent goes to find the leader to give the information.
This latter then resets its computation, spreading the information as in the previous proof. With probability 1,
at some point, all agents will have found the right $Next$ and the leader will then reset for the last time the computation.

We will also use some protocols of  \cite{DBLP:journals/dc/AngluinAE08a}. Even though some parts use
only epidemics, others are trying to detect when something has finally occurred (for example,
detect when some state no longer appears in the population).
When our Las Vegas protocol will run this detection, it will iterate the epidemic part until it
detects the desired fact. In  \cite{DBLP:journals/dc/AngluinAE08a}, those elements are proved
to happen with high probability in a single epidemic. Hence, our expectation will not grow here.

We can prove that this protocol takes at most $O(n^2\log n+n^2\log n+n^3)=O(n^3)$ expected interactions
to reset for the last time the computation.
Then, we add a factor of $n^3$ to the expected number of interactions taken by the Monte Carlo protocol
to make it Las Vegas.

As the Monte Carlo protocol fails with probability at most $n^{-c}$ and that the expected
number of interactions of the Las Vegas protocol is polynomial, the parallel expectation is still polylogarithmic.

\subsection{The Size of the Population}

The purpose of the following section
 is to find a way to compute the size of the population.
As each agent can only contain a finite state, each agent will store one bit,
and the $\log n$ first agents (according to their identifiers) will ultimately have the size
written in binary when you align their bits according to their order.
This way to encode an input size was also used in \cite{Rabie15}.

The protocol uses a sub-protocol that computes the median identifier
of a given subset of agents. Used on the whole population, we get
the first bit of the size (depending on if we have the same number of identifiers
bigger and smaller to this identifier or not). We can then work on half the population.
We iterate the protocol on the new half to get a new bit and a new half.



\begin{theorem}\label{median}
Finding the median identifier can be done in a polylogarithmic number of parallel interactions.
The median identifier is the identifier $Med$ such that:\\
$$|\{id:id\le Med\}|-|\{id:id> Med\}|\in\{0,1\}$$
\end{theorem}
\begin{proof}
We will give an idea here of the protocol. A better description can be found in
the appendix.

The protocol works by dichotomy. It keeps and updates two identifiers $Min$ and $Max$
that bounds the median identifier. Here is a quick description of the steps of the protocol:
\begin{enumerate}
\item We initialize $Min$ and $Max$ by finding through an epidemic the
smallest and the highest identifier present in the population.
\item The leader takes at random an identifier $Cand$ in $]Min,Max[$, by picking the
first identifier in the interval it hears about (spreading the search of such an identifier and the reception
takes two epidemics).
\item The leader performs the predicates $|x_{\le Cand}-x_{> Cand}|=0$,
$|x_{\le Cand}-x_{> Cand}|=1$ and $|x_{\le Cand}-x_{> Cand}|\ge2$,
using protocols from \cite{DBLP:journals/dc/AngluinAE08a} (see Theorem \ref{thfc}),
where $x_{\le Cand}$ is the number of agents with an identifier smaller or equal to $Cand$ and
$x_{>Cand}$ is the number of agents with an identifier higher than $Cand$.
\begin{itemize}
\item If the answer is $True$ for one of the two first predicates, $Cand$ is
the median identifier. The algorithm is over.
\item If the answer for the third predicate is $True$, we have $Min<Cand<Med<Max$.
We replace $Min$ with $Cand$ and go back to Step 2.
\item Else, we know that $Min<Med<Cand<Max$.
We replace $Max$ with $Cand$ and go back to Step 2.
\end{itemize}
\end{enumerate}

We prove in the appendix that there is a probability greater than $\frac14$
to divide by $\frac43$ the number of identifiers in the interval $]Min,Max[$ after
one loop of the algorithm. This permits to conclude that this algorithm will
take $O(\log n)$ expected iterations.

Each iteration using $O(n\log^5n)$ expected interactions, we get that this
protocol is in $CPPL$.
\end{proof}

The previous protocol will be used as a tool to
write the size of the population on the $\log n$ first agents.
It still work on any subset of agents.

\begin{theorem}\label{slotcisuber}
There exists a protocol that writes in binary on the first
$\log n$ agents the size of the population.
\end{theorem}
\begin{proof}
To build this protocol, we first adapt the previous one as follows:
\begin{itemize}
\item The Median protocol can be used on a segment:\\
Instead on working on the whole population, we accept to launch it
with two identifiers $A$ and $B$. We will look on the median identifiers
among those who are in $[A,B]$. 
\item The protocol needs to check if the number of agents
in the segment $[A,B]$
is even or odd. This corresponds to check if $|\{A\le {id}\le Med\}|-|\{B\ge {id}> Med\}|$
is equal to 0 or 1.
\end{itemize}

Each agent stores a bit $Size$ set to 0.
The bits of the size are computed from the right
to the left as follows:
\begin{enumerate}
\item Let $min$ (resp. $max$) be the smallest (resp. higher)
identifier present in the population. We initialize $A$ and $B$ with, respectively,
$min$ and $max$.

We also initialize an identifier $C$ to $min$, it will represent
the cursor pointing to which agent we write the bit 
of the size of the population when it is computed.
\item We compute $Med$, the median agent in $[A,B]$, and write
the parity on the bit $Size$ of agent $C$.
\item We update the identifiers as follows:
$B\leftarrow Med$, $C\leftarrow Next(C)$.
\item If $A\ne B$, we come back to step 2, else the computation is over.
\end{enumerate}

When this protocol is over, we have
$$n=\sum\limits_{i=0}^{\log n}2^iSize_{Next^i(min)}.$$
where $Size_{Next^i(min)}$ is the bit $Size$ of the $(i+1)$th
agent.

The Median protocol will be iterated exactly $\log n$ times.
This concludes the proof.
\end{proof}


\section{Impossibility Results}\label{sec:4}

In this section, we provide two results that motivate the idea that
the population cannot take into consideration precisely the sub-words in the population
(and hence, focus locally on the input).
More precisely, only a polylogarithmic number of agents may know
what there is exactly on their "neighbors". It is supported by the fact that
only a polylogarithmic number of agents will know the identifier next of their own
(Theorem \ref{ap}).

The proof that Linearly Local Languages (see Definition \ref{lll})
are not in $CPPL$ (Theorem~\ref{abstar}) is based
on the fact that there is a pair of consecutive identifiers 
such that, with high probability, the population will not be able to differentiate them,
as these identifiers will not appear in a common interaction during the computation.




\begin{theorem}\label{ap}
Any population protocols needs at least $\Omega\left(n\sqrt n\right)$ expected interactions until
each agent has found its $Next$.
\end{theorem}

The proof is in the appendix.


%
%
%


We bring now another impossibility result. 
We show that Community Protocols cannot link a linear number of consecutive
identifiers in $CPPL$. To prove this, we introduce a new class of languages:

\begin{define}\label{lll}
Let $u=u_1\ldots u_N$ a word of size $N$ and $i<N$. We call $\sigma_i(u)$ the word
$u$ where the $i$th letter is permuted with the next one. More formally, we have:
$$\sigma_i(u)=u_1\ldots u_{i-1}u_{i+1}u_iu_{i+2}\ldots u_N$$

We say that a language $L$ is \motnouv{Linearly Local} if there exists some $\alpha\in]0,1]$
such that, for any $n$, there exists some $u\in L$ and some $I\subset\N$ such that:

$u=u_1\ldots u_N$ with $N\ge n$, $\exists I\subset[1,N-1]$, $|I|\ge\alpha N$ and for all $i\in I$,
$\sigma_i(u)\not\in L$.
\end{define}

These languages are said linearly local as, for any size of input, we can find words
that have a linear number of local regions where a small permutation of letter
leads to a word not in the language.


\begin{theorem}\label{abstar}
There is no linearly local language in $CPPL$.
\end{theorem}

To prove this result, the idea is to prove that for any protocol, and for any $n$,
there exists some $u$ in the language of length at least $n$ and $i\in I$ such that there
is a high enough probability that the protocol acts the same way on the inputs
$u$ and $\sigma_i(u)$.\\

Let $\alpha\le1$, and let $(I_n)_{n\in\N}$ be a sequence such that,
for any $n\in\N$, we have
 $I_n\subset[1,n]$ and $|I_n|\ge\alpha n$.

We work on pairs $(Id_{i},Id_{i+1})_{i\in I_n}$.
We want to prove that, for any $n$,
there is some $i\in I_n$ such as, with high probability, the identifiers $Id_{i}$ and $Id_{i+1}$ never
appeared in the same interaction after any $n$-polylog number of interactions.
In the proof, $Id_i$ meets $Id_{i+1}$ means both identifiers appear in the slots of two interacting agents
when the interaction occurs.

To prove that, we first introduce the 3 following lemmas. The proof of the two first ones can be found in the appendix.
\begin{lemma}\label{Ma}
Let $f$ a $n$-polylog function and let $\alpha>0$.

To each identifier ${id}$, we define the set $E_{id}$ and value $M_{id}$ as\\
$E_{{id}}=\{\text{Agents having had }{id}$ in one of its register after $f(n)$ steps$\}$ and $M_{{id}}=|E_{{id}}|$.

There exists some polylogarithmic function $g$ such that, for $n$ large enough,
after $f(n)$ steps:
$$\E(|\{{id}:M_{{id}}\le g(n)\}|)\ge\left(1-\frac\alpha2\right)n$$
\end{lemma}

With this first result, we deduce that at most a small fraction of
the pairs $(Id_{i},Id_{i+1})$ could have met after 
$n$-polylog number of steps. This means that $Id_{i}$ and $Id_{i+1}$
never appeared in the slots of two agents that interacted together, when they interacted.

\begin{lemma}\label{Ma2}
Let $f$ be a $n$-polylog function. For $n$ big enough, after $f(n)$ steps:
$$\E(|\{i\in I_n:Id_{i}\text{ and }Id_{i+1}\text{ were in a same interaction}\}|)\le \frac34\alpha n$$
\end{lemma}

\begin{lemma}
Let $f$ be a $n$-polylog function. For any $n$ large enough, there exists $i\in I_n$ such as:
$$\Pr(Id_{i}\text{ met }Id_{i+1})\le\frac34$$
\end{lemma}
\begin{proof}
  Let suppose that for any $i$, $\Pr(Id_{i}\text{ met }Id_{i+1})\ge\frac34$.

That implies, $\E(N)=\sum\limits_{i\in\I_n}\Pr(Id_{i}\text{ met }Id_{i+1})\ge\frac34\alpha n$.

This is a direct contradiction of the previous lemma.
\end{proof}

 To prove our
theorem, we need to prove the following proposition:

\begin{proposition}\label{prop:proba}
For any protocol, for any $n$-polylog function $f$,
for any input of size $n$ large enough, there exists some $i\in I_n$ such that the
probability that the identifiers
$Id_{i}$ and $Id_{i+1}$ never appeared on a same interaction
after $f(n)$ steps is greater than $\frac14$.
\end{proposition}

This proposition is a direct corollary of previous lemma. With this proposition,
the proof of Theorem \ref{abstar} can be done as follows:

\begin{proof}
Let $L$ be a linearly local language with parameter $\alpha>0$.
Let $P$ be a protocol computing $L$ in less than $n\log^mn$ expected interactions
for some $m\in\N$.

Let choose $n$ large enough to have the property of Proposition \ref{prop:proba} with 
$f(n)=9n\log^mn$. Let $u$ be a word of size $N\ge n$ such that the corresponding
$I$ has a size greater than $\alpha N$.

We have, from Markov's Inequality that :

$\Pr($number of steps to compute $u$ $\le 9N\log^mN)\ge\frac89$.

It implies that at least $\frac89$ of the sequences of configurations of length $9N\log^mN$
provides the right output.


By applying the previous proposition, we obtain
the existence of some $i\in I$ such that the
probability that the identifiers
$Id_{i}$ and $Id_{i+1}$ never appeared on a same interaction
after $9N\log^mN$ steps is greater than~$\frac14$.

This implies that in at least $\frac14$ of the sequences of configurations of length $9N\log^mN$,
$Id_{i}$ and $Id_{i+1}$ were never in a common interaction.

Hence, if 
$Id_{i}$ and $Id_{i+1}$ never appear on a same interaction,
then $P$ will not see any difference between
the two inputs $u$ and $\sigma_i(u)$.


Between the $\frac89$ of the sequences that provides the right output on
these two inputs, at least $\frac79$ are common (two sequences are here said to be common
if the sequence of the interacting identifiers are equals).

As $\frac79\ge 1-\frac14$, amongst those common sequences, 
some of them does not involve $Id_{i}$ and $Id_{i+1}$ in a same interaction.
As the protocol cannot differentiate those two inputs during those sequences,
it cannot bring the right output.

This provides a contradiction. There is no protocol in $CPPL$ that computes $L$.
%
%
%
%
\end{proof}

\begin{corollary}
The rational language $(ab)^*$, the rational language of words not containing the subword $(ab)$, 
the well-formed parenthesis language
and the palindrome language
are not in $CPPL$.
\end{corollary}
\begin{proof}
For the first language, to each $n$ we can associate $u=(ab)^n$, with $\alpha=1/2$.
Same thing with the third one, replacing $a$ with the opening parenthesis and $b$
with the closing one. For the fourth, $(ab)^na$ works the same way.
Finally, for the second one, $(bac)^n$ and $\alpha=1/3$ works.
\end{proof}

\section{Set Considerations}\label{sec:5}


We provide finally, in this section, set comparisons with $CPPL$.
We first give a large upper bound:

\begin{theorem}\label{inclu1}
$$CPPL\subset NSPACE(n\log n)\cap \bigcup\limits_{k\in\N} SPACE(n\log^kn)$$
\end{theorem}
The result is a combination of Theorem \ref{cpmain} with a lemma proved in the appendix.


We now provide  a class of Turing Machines that computes everything we found
to be computable yet. This class of machines is capable of computing global properties,
through the ability
to work on subsets of agents. It is capable to compute the size of sets of agents.
It can perform any polylogarithmic number of steps of a regular Turing Machine.

This machines are capable of focusing only on a polylogarithmic regions of agents.
It motivates the belief that Community Protocols are not capable of local knowledge
on too much places.

\begin{theorem}\label{tighttheo}
Let $M_T$ a Turing Machine on alphabet $\Gamma$ recognizing the language $L$
having the following restrictions. There exists some $k\in\N$ such that
\begin{itemize}
\item $M_T$ has 4 tapes. The first one is for the input $x$.
\item The space of work is restricted as follows:
\begin{itemize}
\item The first tape uses only the input space of $|x|$ cells.
\item The 2nd and the 3rd use at most a space of $\log|x|$ cells.
\item The 4th uses at most a space of $\log^k|x|$ cells.
\end{itemize}
\item $M_T$ can only do
at most $\log^k|x|$ unitary operations among the following ones:
\begin{enumerate}
\item A regular Turing Machine step.
\item Mark/Unmark the cells that have the symbol $\gamma\in\Gamma$.
\item Write in binary on the 2nd tape the number of marked cells.
\item Go to the cell of the number written on the 3rd tape if this number is 
smaller than $|x|$.
\item Mark/Unmark all the cells left to the pointing head on the first tape.
\item Turn into state $\gamma'$ all the marked cells in state $\gamma\in\Gamma$.
\item Select homogeneously a random number between 1 and the number written on the 3rd tape if this number is 
smaller than $|x|$.
\end{enumerate}
\end{itemize}

Then we have $L\in CPPL$.
\end{theorem}
\begin{proof}
The proof is in the appendix. All the items are proved using previous results.
\end{proof}

\bibliographystyle{plain}

\newpage
\appendix

\section{Proof of Theorem \ref{median}}

Here is a better description of the median identifier protocol.

As said, this algorithm is based on dichotomy.
Let construct the following recursive protocol
that has as input two identifiers $Min$ and $Max$ and tries to find $med$ knowing
that it is in $]Min,Max[$.

We do not provide a full description of the protocol
but all the ideas to build it. We first describe the $d$
 slots of identifiers. Each agent will also store some bits also described bellow.
We then describe the algorithm of the protocol. We finally provide a lemma
that proves that the protocol is in $CPPL$.

Each agent stores four slots of identifiers and three bits:
\begin{itemize}
\item Slot $Leader$ containing the leader's identifier.
\item Slots $Min$ and $Max$ (that will be initialized in steps 0.1 and 0.2, before launching the first time
the dichotomy protocol).
\item Slot $Cand$ that contains the current candidate to be the median identifier.
\item Bit $C$ equals to 1 if and only if the agent's identifier is in $]Min,Max[$.
\item Bit $D$ equals to 1 if and only if $Cand$ is the candidate identifier selected by the leader.
\item Bit $G$ that equals 1 if and only if the agent's identifier is strictly greater
than $Cand$.
\item Bit $S$ that equals 1 if and only if the agent's identifier is smaller or equal
to $Cand$.
\end{itemize}
To an agent with identifier $id$, its slot $Slot$ will be written $Slot_{id}$ and its bit $B$
will be written $B_{id}$.

We will not describe formally all the rules of the protocol
but the essential steps. Here are
of the steps of the protocol (0.1, 0.2 and 0.3 are the initializing steps):
\begin{enumerate}
\item[Step 0.1] We start with a Leader Election using the smallest identifier as a Leader.
Each agent stores the leader's identifier in its slot $Leader$ and in its slot $Min$.
\item[Step 0.2] At the same time, the Leader performs an epidemic,
each agent updating its slot $Max$ each time they see a bigger candidate.

We know that when the epidemic stops, each agent, with high probability,
knows the smallest and the biggest identifier present in the population.
\item[Step 0.3] Every agents that are not $Min$ and $Max$ set their slot
$Cand$ to their own identifier, put their bit $C$ to 1 and their bit $D$ to 0.

$Min$ and $Max$ put $\_$ in slot $Cand$ and set their bit  $C$ and $D$ to 0.
\item[Step 1.] The Leader looks for a candidate to be the medium identifier $Med$. For this, it propagates
an epidemic where each agent such that $Cand=\_$ fills it as soon as it meets
an agent having an identifier with this slot filled (putting it in its own slot).
\item[Step 2.] When the leader finds a candidate, it starts a new epidemic to spread 
the candidate's identifier
 to each agent. When an agent $id_1$ with $D_{id_1}=0$ meets and agent $id_2$ with $D_{id_2}=1$, it copies
 $Cand_{id_2}$ in its slot $Cand_{id_1}$ and switches its own bit $D_{id_1}$.
Then  it updates its bits $G_{id_1}$ and $S_{id_1}$ (according to its relative position to the identifier $Cand_{id_2}$).
\item[Step 3.] The leader launches two protocols described in \cite{DBLP:journals/dc/AngluinAE08a}.
The first one decides if $\sum\limits_{id}(S_{id}-G_{id})\in\{0,1\}$ and 
the second one decides if
$\sum\limits_{{id}}(S_{id}-G_{{id}})\ge2$.

Section \ref{subsectfcpp} recalled that a subtraction can be
performed in $O(\log^3 n)$ expected epidemics.

Note that $\sum\limits_{{id}}S_{{id}}$ corresponds to the number of
agents having an identifier smaller or equal to slot $Cand$, and
 $\sum\limits_{{id}}G_{id}$ corresponds to the number of agents greater.
\item[Step 4.] According to the result, the protocol acts as follows:
\begin{itemize}
\item If $\sum\limits_{{id}}(S_{{id}}-G_{{id}})\in\{0,1\}$, slot $Cand$ corresponds to identifier $Med$,
the protocols ends.
\item If $\sum\limits_{{id}}(S_{{id}}-G_{{id}})\ge2$, we have $Min<Cand<Med<Max$.
The Leader does an epidemic asking each agent to put identifier $Cand$ in their slot $Min$.
Each agent checks if its identifier is smaller than $Cand$. If it is the case, the agent
updates $Cand$ to \_ and switches $C$ to 0, else it updates $Cand$ to its own identifier.

Then the Leader goes back to step 1.
\item Else,  we have $Min<Med<Cand<Max$.
The Leader does an epidemic asking each agent to put $Cand$ in their $Max$.
Each agent checks if its identifier is higher than $Cand$. If it is the case, the agent
updates $Cand$ to \_ and switches $C$ to 0, else it updates $Cand$ to its own identifier.

Then the Leader goes back to step 1.
\end{itemize}
\end{enumerate}

\begin{lemma}
The protocol described above finishes in $O(n\log^3n)$ expected interactions.
\end{lemma}
\begin{proof}
This protocol finishes, as in each step, there is at least one less candidate
(formally, $\sum\limits_{{id}}C_{{id}}$ is strictly decreasing after each
passage in the loop).
Each loop uses a polylogarithmic number of epidemics
(The decision problem defined in \cite{DBLP:journals/dc/AngluinAE08a} uses
$O(\log^2n)$ expected interactions).

Hence, we only need to prove that we have a polylogarithmic number of
loops in expectation.\\

Let $c_k$ be the number of agents such as $C_{{id}}=1$
after the $k$th loop. We have $c_0=n$ and $\forall k$, $c_k>c_{k+1}$.
We will show now that if $c_k>1$, $\Pr\left(c_{k+1}\le\frac34c_k\right)\ge\frac14$.\\

We notice first that each agent having its bit $C$ set to 1
has the same probability to be chosen by the Leader, as each of
these agents act the same way. We will call $Cand$ the identifier
of the selected agent.

We assume that $|\{{id}:{id}\le Med$ \& $C_{{id}}=1\}|\ge|\{{id}:{id}\ge Med$ \& $C_{{id}}=1\}|$
(the case $\le$ is symmetric). We have $\Pr(Cand\le Med)\ge1/2$, as we chose to focus
on the case where the majority of candidates have a smaller identifier than $Med$.

We now suppose that $Cand\le Med$. Let $M$ be the median identifier
of the subset $\{{id}:{id}\le Med$ \& $C_a=1\}$.
We have $\Pr(Cand\ge M|Cand\le Med)\ge1/2$, as $M$ is the median identifier
among those smaller than $Med$.

In the case $Cand\ge M$, we have the following inequality on the number
candidates that will no longer be one:\\
$|\{{id}:{id}\le M$ \& $C_{{id}}=1\}|\ge\frac12|\{{id}:{id}\le Med \& C_{{id}}=1\}|
\ge\frac12\times\frac12|\{{id}:C_{{id}}=1\}|=\frac14c_k$.

Hence, if we have $Cand\ge M$, we have $c_{k+1}= c_{k}-|\{{id}:{id}\le Cand$ \& $C_{{id}}=1\}|\le\frac34c_k$.

The expectation of the number of trials before we are in this case is less than 4
(as we have more than half a chance to be inferior to $Med$ and then
have again more than half a chance to be greater than $M$).

$(3/4)^ic_0\le 1\Leftrightarrow i\ln(3/4)+\ln n\le0\Leftrightarrow i\ge\frac{\ln n}{\ln 4/3}$.

This protocol uses, in expectation, at most $\frac4{\ln4/3}\ln n$ loops
before finding the $Med$ identifier.\\

This protocol uses $O(n\log^3n)$ expected interactions to compute
the median identifier.

\end{proof}

\section{Proof of Theorem \ref{ap}}

Let $T$ be the number of interactions needed in a sequence of configurations $(C_i)_{i\in\N}$
to reach the point where all agents have the right $Next$.
We define, for each identifier ${id}$ and each configuration $C_i$, 
$M_{id}(i)$ as the number of agents, $Previous(id)$ excluded, that had $id$ written in one
of its $d$ slots of identifiers in a configuration $C_j$ with $j\le i$.
We also define $M_{id}=M_{id}(T)$.

We first define the following elements:
\begin{itemize}
\item $E_{{id}}(i)=\{{id}_b\in U:{id}_b\ne Previous({id})$ $\&$ ${id}$ has been in a slot of ${id}_b$ in a configuration $C_j$ with $j\le i\}$. It corresponds to the set of identifiers where $id$ appeared in a slot, excluding $Previous(id)$.
\item $M_{{id}}(i)=|E_{{id}}(i)|$.
\item $E_{id}=E_{id}(T)$ and $M_{{id}}=M_{{id}}(T)$.
\end{itemize}

We first have the following lemma:

\begin{lemma}\label{lemmaroberts}
$$\sum\limits_{{id}\in U}M_{{id}}\le d(n+2T).$$
\end{lemma}

\begin{proof}
From a configuration $C_i$, we define, for ${id}\in U$, $p_{id}(i)\subset U$ the set of
identifiers appearing in the $d$ slots of ${id}$.

Let $L_{{id}}(i)$ be the set of all of the $p_{id}(j)$, with $j\le i$ (we have $L_{{id}}(i)\in\Part(U)$, $\Part(U)$
being the set of subsets of $U$).

Let $N_{{id}}(i)$ be the number of times ${id}$ appears in each $L_{id_b}(i)$ and $N_{{id}}=N_{{id}}(T)$.
More formally:
$$N_{{id}}(i)=\sum\limits_{{id_b}\in U}|\{p:{id}\in p~\&~p\in L_{{id_b}}(i)\}|$$

Finally, let $S(i)=\sum\limits_{{id}\in U}\sum\limits_{p\in L_{{id}}(i)}|p|$ and $S=\sum\limits_{{id}\in U}\sum\limits_{p\in L_{{id}}}|p|$.

We have $S(i)=\sum\limits_{{id}\in U}N_{{id}}(i)$, as each element of each set $p$ is counted
exactly once in one of the $N_{{id}}(i)$.

We also have $N_{{id}}(i)\ge M_{{id}}(i)$, as, if ${id}_b\in E_{{id}}$, there is some $p\in L_{{id}_b}$
such as ${id}_b\in p$. We then also get $N_{{id}}\ge M_{{id}}$.

Let $Z(i)=\sum\limits_{{id}\in U}|L_{{id}}(i)|$ and $Z=\sum\limits_{{id}\in U}|L_{{id}}|$.

We know that for any $p$, $|p|\le d$ (each agent stores at most $d$ identifiers at a same time).
From this, we get  
$$S(i)\le\sum\limits_{{id}\in U}\sum\limits_{p\in L_{{id}}(i)}d=d\sum\limits_{{id}\in U}|L_{{id}}(i)|=dZ(i).$$

We can notice that $Z(0)=n$ and that for any $i$, $Z(i+1)\le Z(i)+2$, as an interaction
can only add at most one element in the $L(i)$ of the interacting identifiers.

From this, we have $Z(i)\le n+2i$ et $Z\le n+2T$.

We finally obtain $\sum\limits_{{id}\in U}M_{{id}}\le\sum\limits_{{id}\in U}N_{{id}}=S\le dZ\le d(n+2T).$
\end{proof}

We consider now $T$ as the random variable corresponding to the end of the arrangement
protocol (i.e. the first time where all agents know their $Next$ identifier),
and let $\E(T)$ be its expectation.\\

\begin{lemma}\label{lemmawatson}
$$\ETf\ge \frac1{\sqrt{20d}}n\sqrt{n}.$$
\end{lemma}
\begin{proof}
We work in the case where $T\le2\ETf$.

Let $p_k$ be the probability that $k+1$ agents found their $Next$,
supposing that $k$ agents have already found their $Next$.

As $Previous$ is bijective and  $\{{id}\text{ without }Next\}\subset\{{id}\in U\}$, we have
$$p_k=\frac1{n(n-1)}\sum\limits_{{id}\text{ without }Next}M_{Previous({id})}\le \frac1{n(n-1)}\sum\limits_{{id}}M_{{id}}.$$

By applying Lemma \ref{lemmaroberts} and the hypothesis $T\le2\ETf$, we obtain:
$$p_k\le \frac d{n(n-1)}(n+4\ETf).$$

We have $\ETf\ge n$, since the quickest way to end the protocol
is by having each agent meeting its $Next$ directly. Hence, $p_k\le\frac d{n(n-1)}5\ETf$.

Hence, $\E(T|T\le2\ETf)=\sum\limits_{k=0}^{n-1}\frac1{p_k}\ge\sum\limits_{k=0}^{n-1}\frac{n(n-1)}{5d\ETf}=\frac{n^2(n-1)}{5d\ETf }$.\\

In the general case, we have:\\
$\E(T)=\Pr(T\le2\ETf)\cdot\E\left(T|T\le2\ETf\right)+\Pr(T>2\ETf)\cdot\E(T|T>2\ETf)$.

Markov's inequality provides the result that:
$\Pr(T\le2\ETf)\ge\frac12$. Hence, we deduce:\\
$\E(T)\ge\Pr(T\le2\ETf)\cdot\E\left(T|T\le2\ETf\right)\ge\frac{n^2(n-1)}{10d\ETf}$.\\

From this, we have $\E(T)^2\ge\frac{n^2(n-1)}{10 d}\ge\frac{n^3}{20 d}$. This provides the final result:\\
$$\ETf\ge \frac1{\sqrt{20d}}n\sqrt{n}$$
\end{proof}

Lemma \ref{lemmawatson} permits to conclude the proof of Theorem \ref{ap}:
the arrangement protocol must take at least $O(n\sqrt n)$ expected interactions.


\section{Proof of Lemma \ref{Ma}}

\begin{proof}
Let $h$ be a function such as $\E(|\{{id}:M_{{id}}\ge h(n)\}|)\ge\frac\alpha2n$.

We suppose now to be in the case where $|\{{id}:M_{{id}}\ge h(n)\}|\ge\frac\alpha2n$.

For each identifier ${id}$, we associate the set $L_{{id}}\subset \Part(U)$ of
the list of identifiers the agent ${id}$ have had
during the $f(n)$ first steps.

Let $S=\sum\limits_{{id}\in U}\sum\limits_{p\in L_{{id}}}|p|$. We want to prove that
$S\ge \frac\alpha2nh(n)$.

Let $N_{id}$ be the number of occurrences of ${id}$ in all the $L_{{id}_b}$.

We have $N_{{id}}=\sum\limits_{{id}_b\in U}|\{p:{id}\in p$ $\&$ $p\in L_{{id}_b}\}|$ and $S=\sum\limits_{{id}\in U}N_{{id}}$.

We have $N_{{id}}\ge M_{{id}}$ (an agent ${id}$ appearing at least once in ${id}_b$ appears at least in
one of its lists). As we supposed $M_{{id}}\ge h(n)$ for at least $\alpha n/2$ agents, we get $S\ge\frac\alpha2 nh(n)$.

Our purpose is to get a lower bound of $Z=\sum\limits_{{id}\in U}|L_{{id}}|$.

We know that for any ${id}\in U$ and $p\in\L_{{id}}$, $|p|\le d$.
Then,\\
$$S\le\sum\limits_{{id}\in U}\sum\limits_{p\in L_{{id}}}d=d\sum\limits_{{id}\in U}|L_{{id}}|=dZ.$$

Hence, $Z\ge \frac \alpha {2d}nh(n)$.

Let $Z_i$ be the value of $Z$ after $i$ steps.

We have $Z_0=n$ and $Z_{i+1}\le Z_i+2$.
Hence, $n+2f(n)\ge Z\ge \frac \alpha{2d}nh(n)$.

From this, we get $h(n)\le \frac{2d}\alpha\left(1+\frac2nf(n)\right)$.

As $f$ is $n$-polylog, $h$ must be polylogarithmic.

If we chose $h$ maximal matching our initial postulate, we get
the expected result:
$$\E(|\{{id}:M_{{id}}\le h(n)+1\}|)\ge\left(1-\frac\alpha2\right) n$$
\end{proof}

\section{Proof of Lemma \ref{Ma2}}

\begin{proof}

For any $j$ and $k$, let $L_{j,k}$ be the random variable that is equal to 1
if identifiers $j$ and $k$ appeared in a same interaction, 0 otherwise.

We will work on the number of pairs that interacted $N=\sum\limits_{i\in I_n}L_{i,i+1}$.
We want to prove that the expectation of this variable is less than $\frac34\alpha n$.

$\E(N)=\sum\limits_{i\in I_n}\E(L_{i,i+1})=\sum\limits_{i\in I_n}\Pr(Id_{i}$ met $Id_{i+1}$ after $f(n)$ steps$)$

From Lemma \ref{Ma}, we deduce that from at least $\frac\alpha2n$ pairs,
each identifier is present on at most $g(n)$ agents, as only at most $\frac\alpha2n$ identifiers appeared on more
than $g(n)$ agents.
Hence, for these pairs,


 $\Pr(Id_{i}$ met $Id_{i+1}$ after one step$)\le{\frac{g(n)(g(n)-1)}2}\cdot{\frac2{n(n-1)}}\le\frac{g(n)^2}{n(n-1)}$.

$\Pr(Id_{i}$ met $Id_{i+1}$ after $f(n)$ steps$)\le1-\left(1-\frac{g(n)^2}{n(n-1)}\right)^{f(n)}$.


Let $\beta>1$. The function $l(x)=x^\beta$ is convex (as $l''(x)=\beta(\beta-1)x^{\beta-2}>0$).
The convexity implies that $l(x)\ge l(y)+l'(y)(x-y)$. With $x=(1-X)$ and $y=1$, we have:

$(1-X)^{\beta}\ge 1+\beta(1-X-1)=1-\beta X$.

With $X=\frac{g(n)^2}{n(n-1)}$ and $\beta=f(n)$ (as for $n$ large enough,
$f(n)$ is always greater than 1), we have the following inequality:

$\left(1-\frac{g(n)^2}{n(n-1)}\right)^{f(n)}\ge 1-\frac{f(n)g(n)^2}{n(n-1)}$.

$\Pr(Id_{i}$ met $Id_{i+1}$ after $f(n)$ steps$)\le1-\left(1-\frac{g(n)^2}{n(n-1)}\right)^{f(n)}\le\frac{f(n)g(n)^2}{n(n-1)}$.

As $f$ is $n$-polylog and $g$ is polylogarithmic, for $n$ large enough, we have
$\frac{f(n)g(n)^2}{n(n-1)}\le \frac14$.

If we sum these probabilities for the half number of pairs we considered, we get
the upper bound of $\frac14\alpha n=\frac\alpha4n$ expected pairs that appeared in a same interaction.

Even if all the pairs of identifiers of the not considered half of $I_n$ met, the upper bound
for the second half provides the result:

$$\E(N)\le\frac\alpha2n+\frac\alpha4n\le\frac34\alpha n$$
\end{proof}

\section{Proof of Theorem \ref{inclu1}}

To prove this theorem, it suffices to prove the following lemma:

\begin{lemma}
For any protocol in $CPPL$, there exists some $K\in\N$ such
as this protocol is in $SPACE(n\log^Kn)$.
\end{lemma}
\begin{proof}
Any configuration can be described on a space $O(n\log n)$ by writing,
for each agent, its state and the list of identifiers (using identifier 1 for the
first,\ldots $n$ written in binary for the last).

Let $f(n)=n\log^kn$ for some $k$ being
 a function bounding by above the number of expected interactions to find
the output. We will prove that the result holds considering $K=k+1$.

The idea is to simulate all
the sequences of configurations of length $3f(n)$
(i.e. sequence of $3f(n)$ interactions). 

With more details: We can write a sequence by providing the order of the pairs
that interacted. Hence, we need a space of $3f(n)\log n$
to encode a sequence on a Turing Machine. Going from one sequence to the next
does not require more space.
We then just simulate the sequence of interactions on the population.
Finding, from a configuration, if it has an output or not does not require more space.

For each possible output, we use a counter initialized at 0.
Each time we computed a sequence, if the configuration does have an output,
we increase the corresponding counter.

When all the sequences have been simulated, we provide the same output
that the one that got the highest counter.

From Markov's inequality, we get that the probability to use more that $3f(n)$
interactions is less than one third. Hence, more than the half of the $3f$ sequences
must provide the right output. The output of our Machine is the same that
the protocol's.

This machine is deterministic and uses a space $O(f(n)\log n)=O(n\log^{k+1}n)$.
%
%
\end{proof}

\section{Proof of Theorem \ref{tighttheo}}

The protocol will write the 2nd and 3rd  tapes on the first $\log n$ agents,
the 4th one on the first $\log^k n$ agents. To initialize it, it will start by a leader election, 
then will compute the size of the population (getting from this $\log n$).

Let prove that the 6 items can be simulated:
\begin{enumerate}
\item It is easy, as the Leader only needs to know on which agent each
lecture head is. When the head goes left, the leader finds with an epidemic the previous identifier
(same when the head goes right with the next identifier).
\item This one can be performed in a single epidemic.
\item It corresponds to the process described in Section \ref{sec:3} performed on the
agents marked.
\item It is a similar process to finding the median identifier, only we work:
\begin{enumerate}
\item We keep two identifiers $Min$ and $Max$.
\item We take a random agent $Cand$ in $]Min,Max[$.
\item We compute the number of agents with identifier smaller than $Cand$.
\item Either it is the right one and it is over, either we have an update of $Min$ or $Max$.
In the second case, we go back to step (b) with the new interval.
\end{enumerate}
We can see that this process will use a logarithmic number of loops in expectation.
\item We perform an epidemic to mark all agents of identifier smaller of
the pointing head's.
\item Again, an epidemic is enough.
\item The leader identifiers the $(x+1)$th identifier in the population. It then spreads an epidemic
to look for an identifier smaller or equal to this one, as in the median candidate process. As in the
median process, each identifier has the same probability to be selected.

After than, the leader needs to count how many identifiers are smaller to this one (his own excluded).
The given number will be chosen in the right interval according to a homogeneous distribution.
\end{enumerate}

\end{document}